\documentclass[11pt]{article}
\def \fullpaper {}

\ifdefined\stoc
\else
\usepackage{amsthm}
\usepackage{fullpage}
\fi

\usepackage{times}
\usepackage{subfigure}
\usepackage{color}
\usepackage{graphicx,amssymb,amsmath}
\usepackage{xspace}
\usepackage{algorithm}
\usepackage{algorithmic}

\newcommand{\commentout}[1]{}
\newcommand{\eat}[1]{}

\newcommand{\calD}{{\mathcal D}}

\newcommand{\mI}{{\mathcal I}}

\newcommand{\calF}{{\mathcal F}}

\newcommand{\Prob}{{\operatorname{Pr}}}
\newcommand{\Exp}{{\mathbb{E}}}

\newcommand{\opt}{\mathsf{OPT}}
\newcommand{\sol}{\mathsf{SOL}}

\newcommand{\ball}{\mathsf{Ball}}

\newcommand{\bw}{\mathbf{w}}

\newcommand{\hx}{\widehat{x}}
\newcommand{\hy}{\widehat{y}}
\newcommand{\hz}{\widehat{z}}

\newcommand{\ty}{\widetilde{y}}

\newcommand{\dist}{\mathrm{d}}

\newcommand{\ftfl}{$\mathsf{FTFL}$}
\newcommand{\ftflc}{$\mathsf{FTFL\text{-}Convex}$}

\newcommand{\ftm}{$\mathsf{FTMed}$}
\newcommand{\uniftm}{$\mathsf{Uni\text{-}FTMed}$}

\renewcommand{\d}{{\sf d}}

\newcommand{\red}[1]{\textcolor{red}{#1}}

\newtheorem{theorem}{Theorem}
\newtheorem{lemma}{Lemma}

\newtheorem{observation}{Observation}

\newtheorem{fact}{Fact}

\newcommand{\set}[1]{\left\{#1\right\}}
\newcommand{\card}[1]{\left|#1\right|}
\newcommand{\floor}[1]{\left\lfloor#1\right\rfloor}
\newcommand{\ceil}[1]{\left\lceil#1\right\rceil}

\newcommand{\mU}{\mathcal U}

\newcommand{\mB}{\mathcal B}

\newcommand{\queue}{\mathsf{queue}}

\newcommand{\dav}{d_{\mathsf{av}}}
\newcommand{\dmax}{d_{\max}}

\usepackage{eqparbox,array}

\title{A Constant Factor Approximation Algorithm for Fault-Tolerant $k$-Median}
\author{Mohammadtaghi Hajiaghayi\thanks{Supported in part by NSF CAREER award 1053605, ONR YIP award N000141110662, DARPA/AFRL award FA8650-11-
1-7162, a Google Faculty Research Award, and a University of Maryland Research and Scholarship Award (RASA). Department of Computer Science, University of Maryland at College Park, USA. The author is also with AT\&T Labs research. Email: \textrm{hajiagha@cs.umd.edu}.},
~ Wei Hu\thanks{Institute for Interdisciplinary Information Sciences
Tsinghua University, China. Email: \textrm{huwei9527@gmail.com}},
~ Jian Li\thanks{Institute for Interdisciplinary Information Sciences
Tsinghua University, China. Email: \textrm{lijian83@mail.tsinghua.edu.cn}.},
~ Shi Li \thanks{Department of Computer Science, Princeton University. Email: \textrm{shili@cs.princeton.edu}}
~and~ Barna Saha\thanks{AT\&T Research Laboratory, Florham Park, NJ 07932. Email: \textrm{barna@research.att.com}. }}

\date{}
\title{A Constant Factor Approximation Algorithm for Fault-Tolerant k-Median}

\begin{document}


\maketitle

\ifdefined\fullpaper\else
\documentclass[11pt]{article}

\begin{document}
\fi

\begin{abstract}
In this paper, we consider the fault-tolerant $k$-median problem and give the \emph{first} constant factor approximation
algorithm for it. In the fault-tolerant generalization of classical $k$-median problem, each client $j$ needs to be assigned to at least $r_j \ge 1$ distinct open facilities. The service cost of $j$ is the sum of its distances to the $r_j$ facilities, and the $k$-median
constraint restricts the number of open facilities to at most $k$. Previously, a constant factor was known only for the special case when
all $r_j$s are the same, and a logarithmic approximation ratio for the general case.
In addition, we present the first polynomial time algorithm for the fault-tolerant
$k$-median problem on a path or a HST by showing that the corresponding LP
always has an integral optimal solution.

We also consider the fault-tolerant facility location problem, where
the service cost of $j$ can be a weighted sum of its distance to the
$r_j$ facilities. We give a simple constant factor approximation algorithm, generalizing
several previous results which only work for nonincreasing weight vectors.
\end{abstract}
\newpage
\clearpage
\setcounter{page}{1}
\section{Introduction}

The $k$-median problem is one of the central problems in approximation algorithms and operation research.
The most basic version of the $k$-median problem is defined as follows.
We are given a set of facilities $\calF$ and a set of demands (or clients) $\calD$
in a metric space.
We can open at most $k$ facilities, and then assign each client $j$ to the opened facility that is closest
to it. Assigning demand $j$ to facility $i$ incurs an assignment cost of $\dist(i,j)$,
where $\dist(i,j)$ is the distance between $i$ and $j$.
 Our goal is to choose at most $k$ facilities so that the sum of the assignment costs is minimized.
Lin and Vitter \cite{lin1992approximation} gave a polynomial-time algorithm
that, for any $\epsilon>0$, finds a solution of cost no more than $2+\epsilon$ times the optimum, while
using at most $(1+\epsilon)k$ facilities.
The first non-trivial approximation algorithm that produces
a feasible solution (i.e., open at most $k$ facilities) achieves a logarithmic approximation ratio
by combining the metric embedding results \cite{bartal1998approximating,JSK2003treeembed} and
the fact that $k$-median can be solved in polynomial time in a tree metric.
Charikar, Guha, Tardos and Shmoys~\cite{charikar2002constant} gave
the first constant factor approximation algorithm using LP rounding.
This was improved by a series of papers~\cite{CharikarGuha2005costscale, Jain2003greedy, arya2001local, CL11}
and the current best approximation ratio is $1+\sqrt{3}+\epsilon$ for any $\epsilon>0$
via pseudo approximation~\cite{Li:2013}.
For the {\em fault tolerant} version of $k$-median (\ftm),
each client $j$ needs to be assigned to at least $r_j \ge 1$ distinct open facilities.
The service cost of $j$ is the sum of its distances to the $r_j$ facilities.
A special case of \ftm\ is when all the $r_j$s are the same. We call such instance as {\it uniform} \ftm\
(denoted by \uniftm).
For \uniftm, Swamy and Shmoys \cite{swamy2008fault} developed a $4$-approximation using
the Lagrangian relaxation technique. However, their technique does not
work when $r_j$s are not same, even when $r_j$s are either $1$ or $2$.
For general \ftm, where $r_j$s can be non-uniform, the best known result is a logarithmic factor approximation algorithm~\cite{anthony2008plant}.

In the closely related uncapacitated facility location problem (UFL),
there is a facility opening cost $f_i$ for each facility $i$ and
our objective is to minimize the sum of the facility opening cost
and the total assignment cost.
The first constant factor approximation algorithm for UFL was given by Shmoys, Tardos and Aardal~\cite{shmoys1997approximation},
using the filtering technique of Lin and Vitter \cite{lin1992approximations}.
Subsequently, a variety of techniques in approximation algorithms
has been successfully applied to UFL ( see e.g., \cite{Chudak1998improvedapproximation, Jain2001primal, arya2001local, Archer2003lagrangianrelaxation, Jain2003greedy, Chudak2004, CharikarGuha2005costscale, Li2011}).
The current best approximation ratio is 1.488 by Li~\cite{Li2011},
which is quite close to the best known inapproximability bound of 1.463 due to Guha and Khuller~\cite{guha1998greedy}.
In this paper, we study the {\em  fault-tolerant} version of UFL where
each client $j$ needs to be assigned to at least $r_j \ge 1$ distinct open facilities.
Client $j$ is associated with a weight vector $\bw_j=\{w_j^{(1)},w_j^{(2)},\ldots, w_j^{r_j}\}$.
The service cost of $j$ is the weighted sum of its distances to the $r_j$ facilities,
i.e., $\sum_i w_j^{(i)} \dist(h_i,j)$ where $h_i$ is the $i$th closest open facility.
It models the situation where each client needs one or more ``backup'' facilities
in case its closest facility fails.
The fault-tolerant facility location (\ftfl) is a generalization of UFL in which $r_j = 1$ for each client $j$.
\ftfl\ with nonincreasing weight vectors ($w_j^{(1)}\geq w_j^{(2)}\geq \ldots$ for each client $j$)
has been studied extensively.
Jain and Vazirani gave a primal-dual based algorithm
achieving a logarithmic approximation factor \cite{kamal2003approximation}.
The first constant factor approximation algorithm
with a factor of $2.408$ is due to Guha, Meyerson and Munagala~\cite{guha2003constant}.
This was later improved to $2.076$ by Swamy and Shmoys~\cite{swamy2008fault}
and $1.7245$ by Byrka, Srinivasan and Swamy~\cite{byrka2010fault}, which is currently
the best known ratio.
However, nothing is known for \ftfl\ with general positive weight vectors.
Measuring service cost using general weight vectors is often a natural choice.
For example, in the fault-tolerant $k$-center problem \cite{khuller1997fault, chaudhuri1998p},
the service cost of a client is chosen to be its distance to the $r$th closest facility
(this corresponds to the weight vector $(w^{(1)}_j=0, \ldots,w^{(r-1)}_j=0 , w^{(r)}_j=1, w^{(r+1)}_j=0, \ldots)$).
Further consider the following application in a wireless sensor network.
We need to place hotspots (facilities) to provide wireless services for a designated area.
Each hotspot may fail independently with probability $p$ at every time slot.
Each client is a sensor that needs to communicate with one hotspot.
To ensure that the communication succeeds with probability at least $1-\delta$ at each time slot,
the transmission radius (fixed all the time) of the client needs to
be the distance from the client to its $\lceil \log_p \delta \rceil$th closest hotspot.
If the communication cost of a client scales linearly with its transmission radius,
the problem is exactly \ftfl\ with weight vectors of the form $(0,\ldots, 0,1,0,\ldots)$.

\subsection{Our Results}


Our main result is a constant factor approximation algorithm for general \ftm.
The current best approximation algorithm for
general \ftm\ achieves a logarithmic approximation ratio~\cite{anthony2008plant}.
Note that no constant factor approximation algorithm is known even for the case where
the demands are either 1 or 2 and no previous techniques for $k$-median or uniform \ftm\
\cite{charikar2002constant, arya2001local, jain1999primal, CL11, swamy2008fault}
seems to be generalizable easily to this case.
Our algorithm is built on solving the natural linear programming (LP) relaxation of \ftm. Rounding is involved and proceeds through stages.
First, based on the LP solution, we classify the clients into {\em safe} and {\em dangerous}.
The safe clients are those whose distance to the furthest fractional facility assigned to it
can be bounded by a constant factor of the connection cost defined by the LP solution
(for the precise definition, see Section~\ref{sec:flkmed}).
Handling such clients is easy and well understood in recent literature
on the fault-tolerant facility location problem \cite{swamy2008fault,byrka2010fault,Yan12}.
In fact, in the fault-tolerant facility location problem, by scaling up the facility variables
by a constant factor, one can transform all clients to safe, making it easy to approximate.
However, in \ftm, we can not scale the facility variables
since scaling would violate the constraint that we can open at most $k$ facilities.

Next, we apply the adaptive clustering algorithm in \cite{Yan12}
to produce a family of disjoint sets of facilities that we call \emph{bundles}.
However in \cite{Yan12}, one can select multiple copies of the same facility. In order to avoid that,
we need to keep a new mapping.
In the rounding step, we ensure that each bundle contains exactly 1 open facility
by randomly selecting an open facility inside it
(according to the probabilities suggested by the LP),
and we can show that the expected connection cost of a safe client is bounded by a constant times its connection cost in the LP solution.
On the other hand, handling the dangerous clients is significantly challenging and requires new techniques.

We judiciously create a family $\{B_j\}$ of facility sets for each client $j$ choosing from the fractionally open facilities serving $j$ such that $B_j$
is {\em almost laminar}, that is the two sets are either nearly disjoint, or one is almost contained in the other.
This becomes technically challenging primarily for the fact that demands among the clients could be highly skewed.
Once we have such a structure, further refinements through filtering and other manipulations,
lead to a laminar family of sets of facilities that have the nice property
of $y(B_j)$ being very close to $r_j$.
Here $y(B_j)$ is the expected number of fractional facilities in $B_j$.
In the randomized rounding step,
in addition to guaranteeing every bundle contains exactly 1 facility,
we can also guarantee that every set in the laminar family contains
either $\floor{y(B_j)}$ or $\ceil{y(B_j)}$ open facilities.
Since $y(B_j)$ is close to $r_j$, the rounding procedure opens $r_j$ facilities in $B_j$ with high probability
and this suffices to show a constant approximation for the expected service cost of $j$.


As our second result, we show there is a polynomial time algorithm that can exactly solve general \ftm\ in a line metric.
Unlike the ordinary $k$-median problem on a line, which can be easily solved in polynomial time by dynamic programming,
it is unclear how to generalize the dynamic program to \ftm\ (either uniform or non-uniform).
Our algorithm is in fact based on linear program. We show that the LP always has an optimal solution that is integral.
We rewrite the LP based on any (fractional) optimal solution and show the new LP matrix is totally
unimodular. A similar argument can be used to show that the LP of general \ftm\ on a hierarchically well separated tree (HST)
also has an integral optimal solution.
This improves the result in \cite{charikar1998rounding} where they showed that the integrality gap of the $k$-median LP
on HSTs is at most 2.
\footnote{It is well known that $k$-median on trees
can be solved in polynomial time by combinatorial methods (e.g., \cite{tamir1996pn}).
}

We also consider the fault tolerant version
of the facility location problem (\ftfl) where the service cost
of a client is a weighted sum of the distances to the closest open facility,
the 2nd closest open facility and so on.
Our main result for this problem is a simple constant factor approximation algorithm for \ftfl\
with a general weight vector for each client.
This generalizes several previous results \cite{guha2003constant,swamy2008fault,byrka2010fault},
where the weight vectors are nonincreasing.
For general weight vectors, the most commonly used ILP formulation work does not hold since the optimal integral LP solution may not correspond to a feasible solution. To remedy this, we use an extension of the ILP formulation for facility location proposed by Kolen and Tamir~\cite{kolen1990covering}. However, one can easily construct an example where the LP relaxation for this formulation has an unbounded integrality gap (see Section \ref{sec:ftfl}). Our approach is based on formulating a strengthened LP relaxation for the problem by adding ``knapsack cover constraints" \cite{carr2000strengthening,Bansal2010GeSSC}.

\subsection{Other Related Work}

Facility location and $k$-median are central problems in approximation algorithms.
Many variants and generalizations have been studied extensively in the literature,
including capacitated facility location \cite{pal2001facility, levi2004lp, svitkina2010lower}
and $k$-median \cite{chuzhoy2005approximating}, multilevel facility location \cite{aardal19993},
universal facility location \cite{mahdian2003universal,li2011generalized},
matroid median \cite{hajiaghayi2010budgeted,krishnaswamy2011matroid,CL11},
knapsack median \cite{kumar2012constant,CL11}, just to name a few.
A closely related problem is the fault-tolerant $k$-center problem
which has also been studied
and constant factor approximation algorithms are known for several of its variants
\cite{khuller1997fault,chaudhuri1998p}.
Recently, Yan and Chrobak
studied the fault-tolerant facility placement problem
which is almost the same as \ftfl\ except that
we can open more than one copies of a facility
and they gave a constant factor approximation algorithm based on LP rounding \cite{Yan12}.

\ifdefined\fullpaper\else
\bibliographystyle{alpha}
\bibliography{cluster}
\end{document}
\fi

\ifdefined\fullpaper\else
\documentclass[11pt]{article}

\begin{document}
\fi


\section{Fault Tolerant {\large \em k}-Median}
\label{sec:flkmed}


We use $\mI = \left(k, F, C, d, \set{r_j}_{j \in C}\right)$ to denote a \ftm\ instance.  In the instance, $k\geq 1$ is an integer, $F$ is the set of facilities, $C$ is the set of clients, $d$ is a metric over $F \cup C$ and $r_j \in [R]$ is the requirement of $j$.  The solution of $\mI$ is a set $S$ of $k$ facilities from $F$ and its cost is the sum, over all clients $j \in C$, of the total distance from $j$ to its closest $r_j$ facilities in $S$.

The following is the natural LP relaxation for the \ftm:
\ifdefined \stoc
\begin{align}
\text{minimize} \quad \sum_{j \in C}\sum_{i \in F}d(j, i) x_{i, j}\quad\quad\quad \label{lp:ftm}
\end{align}
\vspace*{-15pt}
\begin{align*}
\text{subject to }& \quad y_i - x_{i,j} \geq 0 \qquad \forall i \in F, \  j \in C \\
& \quad\sum_{i \in F} x_{i, j} = r_j \qquad \forall j \in C \\
& \quad\sum_{i \in F} y_i \leq k \qquad  \\
& \quad x_{i, j}, y_i \in [0, 1]  \quad \forall i \in F,\  j \in C
\end{align*}
\else
\begin{align}
\min \quad \sum_{j \in C}\sum_{i \in F}d(j, i) x_{i, j} \label{lp:ftm}
\end{align}
\vspace*{-15pt}
\begin{alignat*}{4}
y_i - x_{i,j} &\geq 0 &\qquad \forall i \in F, &\  j \in C &\qquad\qquad
\sum_{i \in F} x_{i, j} &= r_i &\qquad \forall j \in C &\\
\sum_{i \in F} y_i &= k &\qquad & &\qquad\qquad x_{i, j}, y_i &\in [0, 1] & \quad \forall i \in F,&\  j \in C
\end{alignat*}
\fi


Throughout the paper, we let $y$ denote the $y$-vector obtained by solving the above LP. For a subset $S \subseteq F$ of facilities, define the \emph{volume} of $S$ to be $y(S) := \sum_{i \in S}y_i$. W.l.o.g., we assume $y(F) = k$.

We can assume $y_i \leq 1$ and $x_{i,j} \in \set{0, y_i}$ by the following splitting operation. Consider a facility $i$ and a client $j$ such that $x_{ij}< y_i$. We replace $i$ with two facilities $i_1,i_2$ and let $y_{i_1}=x_{i_1j}=x_{ij}, y_{i_2}=y_i - x_{i,j}, x_{i_2j}=0$. Of course, when we make such clones of a facility, we can only open one of them.

Instead of using $(y, x)$, we use $\left(\set{y_i}_{i \in F}, \set{F_j}_{j \in C}, g\right)$ to denote an LP solution, where $F_j \subseteq F$ and $y(F_j) = r_j$ for every $j \in C$, and $g$ shall be defined later.  In this solution, $y_i$ indicates whether to open the facility $i$. We assume $0 < y_i \leq 1$ for every $i \in F$. Then $i \in F_j$ if and only if $x_{i,j} = y_i$.   We also assume $F_j$ contains the closest $r_j$ volume of facilities to $j$. That is,  for any $j \in C, i \in F_j, i' \notin F_j$, we have $d(j, i) \leq d(j, i')$.
For some non-empty set $S \subseteq F$, let
$$\dav(j, S) = \frac{\sum_{i \in S}d(j, i)y_i}{y(S)}$$
be the average distance from $j$ to $S$.
Let $\dmax(j, S)$ be the maximum distance from $j$ to any node in $S$,
i.e., $\max_{i \in S}d(j, i)$.

Notice that we can alway split a facility $i$ into $2$ facility $i'$ and $i''$ with $y_i = y_{i'} + y_{i''}$ arbitrarily (replace any $F_j \ni i$ with $F_j \setminus \set{i} \cup \set{i', i''}$) without changing the value of the LP solution.
This turns out to be convenient in the following scenario.
Suppose we are given a sequence of facilities $(i_1, i_2, \cdots, i_m)$ such that $\sum_{s=1}^m y_{i_s} \geq r$.
We are interested in the integer $t$ such that $\sum_{s=1}^{t-1} y_{i_s} < r$ and $\sum_{s=1}^{t} y_{i_s} \geq r$. If $\sum_{s = 1}^t y_{i_s} > r$, we can split $i_t$ into two facilities $i'$ and $i''$ with $y_{i'} = r - \sum_{s=1}^{t-1} y_{i_s}$ and $y_{i''} = \sum_{s=1}^ty_{i_s} - r$. By splitting, we assume we can always find the integer $t$ such that $\sum_{s=1}^t y_{i_s}$ is exactly $r$.
Let $j\in C$ be a client and $S$ be a set of facilities such that $y(S) \geq r$. Sort the facilities of $S$ according to their distances to $j$, from the closest to the furthest.  Let $s$ (resp. $t$) be the integer such that the first $s$ (resp. $t$)  facilities in the order has volume exactly $r-1$ (resp. $r$). Then, $S'$ contains the $p$-th facility in the sequence for every $p$ from $s+1$ to $t$.  So $y(S') = 1$. If $y$ is an integral solution, $S'$ would correspond to the $r$-th closest facility to $j$. Define $\dav^{r}(j, S) = \dav(j, S')$ and $\dmax^{r}(j, S) = \dmax(j, S')$ where $S'$ is the following set.  

We observe some simple yet useful facts. Let $j \in C$ be a client and $S$ be a set of facilities with $y(S) = r$ for some integer $r$. Then, we have that

\begin{enumerate}
\item $\dav^t(j, S) \leq \dmax^t(j, S) \quad \forall t \in [r]$,
\item $\dmax^t(j, S) \leq \dav^{t+1}(j, S) \quad \forall t \in [r-1]$,
\item $\dav(j, S) = \frac{1}{r}\sum_{t=1}^r \dav^t(j, S)$.
\end{enumerate}

For ease of notation, we omit the second parameter of $\dav$ and $\dmax$ if it is $F_j$.
That is, we let $\dav(j) = \dav(j, F_j), \dmax(j) = \dmax(j, F_j), \dav^r(j) = \dav^r(j, F_j)$ and $\dmax^r(j) = \dmax^r(j, F_j)$.

In several steps mentioned above, we may split one facility into several copies.
In the rounding step, to avoid opening more than one copies for each facility,
we need to keep a mapping $g$ where $g(i)$ indicates the original facility co-located with $i$
from which $i$ is split. $g(i) = i$ if $i$ itself is the original facility.
Thus, $d(i, g(i)) = 0$. Keep in mind that we need to make sure in the rounding step that at most 1 facility is open
in $g^{-1}(i) :=\set{i' \in F : g(i') = i}$ for any $i \in F$.

The high level idea of our algorithm is as follows.
We solve LP \eqref{lp:ftm} to obtain a fractional solution $\left(\set{y_i}_{i \in F}, \set{F_j}_{j \in C}, g\right)$. Our goal is to output a random set $S\subseteq F$ of size $k$ such that the expected connection cost of $j$
is $O(r_j\dav(j))$ for each client $j$.
We first use the adaptive clustering algorithm of \cite{Yan12}
to construct a family $\mU$ of disjoint sets of volume 1.
If we randomly open 1 facility for each set $U \in \mU$, we can show that
the expected connection cost of each client $j \in C$ is $O(1)r_j\dav(j) + \dmax(j)$.
This can handle the clients $j$  with small $\dmax(j)/(r_j\dav(j))$ (which we call {\em safe clients}).

The remaining task is to handle the dangerous clients, i.e., the clients with a large
 $\dmax(j)/\dav^{r_j}(j)$ value  (the exact definition will appear later).
We first apply a filtering step to select a subset $D'$ of dangerous clients.
For each $j\in D'$, we create a set $B'_j$ of facilities such that
the set family $\mB = \{B'_j : j\in D'\}$ is laminar.
Using the laminar family $\mB$, we design a  process to output a random set $S$ of facilities so that
(1) at most 1 facility is open inside $g^{-1}(i)$ for any $i \in F$,
(2) each facility $i$ is open with probability exactly $y_i$;
(3) exactly 1 facility in each $U \in \mU$ is open and
(4) we open either $\floor{y(B'_j)}$ or $\ceil{y(B'_j)}$ facilities inside each $B'_j \in \mB$.
With these properties, we can prove the constant approximation for \ftm.

The remainder of this section is organized as follows. We show how to construct $\mU$ and
$\mB$ respectively in Section ~\ref{subsec:review-yan} and \ref{subsec:creating-B}.
Then, we show how to round the fractional solution based on $\mU$ and $\mB$ in Section~\ref{subset:rounding}.
Finally, we prove the constant approximation ratio in section~\ref{subsec:proof-of-constant}.

\subsection{Construction of the Family {\large $\mU$}}
\label{subsec:review-yan}

Given a $k$-median instance defined by $k, F, C, d, \set{r_j}_{j \in C}$ and a fractional solution $(\set{y_i}_{i \in F}, \set{F_j \subseteq F}_{j \in C})$ to the instance, the algorithm of  \cite{Yan12} outputs a family $\mU$ of disjoint sets of volume 1, which we call \emph{bundles}, as well as a set $\set{U_{j, t}}_{t \in [r_j]}$ of $r_j$ different bundles from $\mU$ for each $j \in C$.  The algorithm is described in Algorithm~\ref{alg:create-bundles}.

\ifdefined \stoc
\begin{algorithm*}[t]
\else
\begin{algorithm}[t]
\fi

\caption{Create bundles}
\label{alg:create-bundles}
\begin{algorithmic}[1]
\REQUIRE{A FT-$k$-median instance $\mI = \left(k, F, C, d, \set{r_j}_{j \in C}\right)$ and a fractional solution $\left(\set{y_i}_{i \in F}, \set{F_j}_{j \in C}, g\right)$ to $\mI$}\;
\ENSURE{A family $\mU$ of disjoint bundles, and a set $\set{U_{j,t}}_{t\in [r_j]}$ of $r_j$ different bundles from $\mU$ for each $j \in C$}\;
\STATE $\mU \leftarrow \emptyset$,  $F'_j \leftarrow F_j$ and $\queue_j \leftarrow \emptyset$ for every client $j \in C$;
\STATE \textbf{While} there exists a client $j$ such that the length of $\queue_j$ is smaller than $r_j$
\STATE \hspace{\algorithmicindent} Select such a client $j$ with the minimum $\dav^1(j, F'_j) + \dmax^1(j, F'_j)$;
\STATE \hspace{\algorithmicindent} Let $U \subseteq F'_j$ be the 1 volume of facilities such that $\dav^1(j, F'_j) = \dav(j, U)$ and $\dmax^1(j, F'_j) = \dmax(j, U)$;
\COMMENT one might clone facilities in obtaining the set $U$ and $g$ is updated suitably to reflect this.
\label{STATE:select-j-and-B}
\STATE \hspace{\algorithmicindent} \textbf{If} there exists a bundle $U' \in \mU$ such that $U' \cap U \neq \emptyset$
\STATE \hspace{\algorithmicindent} \textbf {then} add $U'$ to the $\queue_j$ and remove $U' \cap U$ from $F'_j$;
\label{STATE:case-intersection}
\STATE \hspace{\algorithmicindent} \textbf {else} add $U$ to $\mU$, add $U$ to $\queue_j$, and remove $U$ from $F'_j$;
\label{STATE:case-no-intersection}
\RETURN $\mU$ and $\set{U_{j, t}}_{j \in C, t \in [r_j]}$, where $U_{j,t}$ is the $t$-th bundle in $\queue_j$.
\end{algorithmic}

\ifdefined \stoc
\end{algorithm*}
\else
\end{algorithm}
\fi

If some $U$ is added to $\mU$ at Line~\ref{STATE:case-no-intersection} of Algorithm~\ref{alg:create-bundles}, we say the \emph{creator} of $U$ is $j$.  It is easy to see that the bundles in $\mU$ are mutually disjoint. Moreover, for any $j \in C$, the $r_j$ bundles added to $\queue_j$ are all different, since every time we add a bundle $U$ to the $\queue_j$, we removed $\mU \cap F'_j$ from $F_j$.

\begin{lemma}
\label{lemma:close-to-bundles}
For any client $j \in C$, for any $r \in [r_j]$, we have $\dav(j, U_{j,r}) \leq 2\dmax^{r}(j) + \dav^{r}(j)$.
\end{lemma}

\begin{proof}
We prove the following statement:  when the length of $\queue_j$ is $r-1$, we have $\dav^1(j, F'_j) \leq\dav^{r}(j)$ and $ \dmax^1(j, F'_j) \leq \dmax^r(j)$.  Notice that we only remove facilities from $F'_j$ if we added some set $B$ to $\queue_j$.  Moreover, we remove at most 1 volume of facilities from $F'_j$.  Thus, when the length of $\queue_j$ is $r_1$, we removed in total at most $r-1$ volume of facilities from $F'_j$.  It is easy to see that in order to maximize $\dav^1(j, F'_j)$ ($\dmax^1(j, F'_j)$, resp.), it is the best to remove from $F'_j$ the $r-1$ volume of closest facilities of $j$, in which case we have $\dav^1(j, F'_j) = \dav^r(j)$($\dmax^1(j, F_j) = \dmax^r(j)$, resp.). Thus, we proved the statement.

Suppose now the length of $\queue_j$ is $r-1$. Clearly, the volume of $F'_j$ is at least 1. Consider the next time when we selected this client $j$ and the correspondent $U$ at Line~\ref{STATE:select-j-and-B}.  We know $\dav(j, U) \leq \dav^r(j)$ and $ \dmax(j, U) \leq  \dmax^r(j)$. If there is a $U' \in \mU$ such that $U' \cap U \neq \emptyset$, let $j'$ be the creator of $U'$. Then, we have $\dav(j', U') + \dmax(j', U') \leq \dav(j, U) + \dmax(j,U)$, since we selected $j'$ and $U'$ before we selected $j$ and $U$.  Thus, $d(j, j') \leq \dmax(j, U) + \dmax(j', U')$ and
\ifdefined\stoc
\begin{align*}
\dav(j, U') &\leq d(j, j') + \dav(j', U') \\
&\leq \dmax(j,U) + \dmax(j',U') + \dav(j', U') \\
&\leq 2\dmax(j,U) + \dav(j,U),
\end{align*}
\else
\[\dav(j, U') \leq d(j, j') + \dav(j', U') \leq \dmax(j,U) + \dmax(j',U') + \dav(j', U') \leq 2\dmax(j,U) + \dav(j,U),
\]
\fi
which is at most $2\dmax^r(j) + \dav^r(j)$.

If such $U'$ does not exist, we added $U$ to $\mU$ and $\queue_j$ at Line~\ref{STATE:case-no-intersection}, we have $\dav(j, U) \leq \dav^r(j)$.
\end{proof}

\subsection{Construction of the laminar Family {\large $\mB$}}
\label{subsec:creating-B}

We say a client $j \in C$ is \emph{dangerous} if
\[\dmax(j) \geq 45\dav^{r_j}(j).\]
The rest of clients are {\em safe}. Let $D$ denote the set of dangerous clients. In this section, we first apply a filtering phase to obtain a subset $D' \subseteq D$ of dangerous clients. Then, for each $j \in D'$ we select a set $B'_j \subseteq F_j$ of facilities so that $\mB = \{B'_j : j \in D'\}$ form a laminar family.

\paragraph{Filtering:} We say two distinct dangerous clients $j, j' \in D$ \emph{conflict} if  $r_j = r_{j'}$ and
\[
d(j, j') \leq 6\max\set{\dav(j), \dav(j')}.\label{equ:filtering-condition}
\]

In the filtering phase, we select a subset $D'\subseteq D$ of dangerous clients such that no two clients in $D'$ conflict each other.  Algorithm~\ref{alg:filtering} describes the filtering process.

\begin{algorithm}
\caption{Filtering}
\label{alg:filtering}
\begin{algorithmic}[1]
\STATE $D' \leftarrow \emptyset$;
\STATE \textbf{For} $r \leftarrow 1$ to $R$ \textbf{do}
\STATE \hspace{\algorithmicindent} $J = \set{j \in D : r_j = r}$;
\STATE \hspace{\algorithmicindent}\textbf{While} $J \neq \emptyset$ \textbf{do}
\STATE \hspace{\algorithmicindent}\hspace{\algorithmicindent}Let $j$ be the client in $J$ with the minimum $\dav(j)$;
\STATE \hspace{\algorithmicindent}\hspace{\algorithmicindent}Let $J'$ be the set of clients in $J$ that conflict $j$;
\STATE \hspace{\algorithmicindent}\hspace{\algorithmicindent}Let $J\leftarrow J \setminus J' \setminus \set{j}$ and $D' \leftarrow D' \cup \set{j}$;
\RETURN $D'$.
\end{algorithmic}
\end{algorithm}

\begin{fact}
If $j \in D \setminus D'$, then there must be a client $j' \in D'$ such that $r_{j'} = r_j$, $\dav(j') \leq \dav(j)$ and $d(j, j') \leq 6\dav(j)$.
\end{fact}


\begin{figure*}[t]
    \begin{center}
    \includegraphics[width=0.8\linewidth]{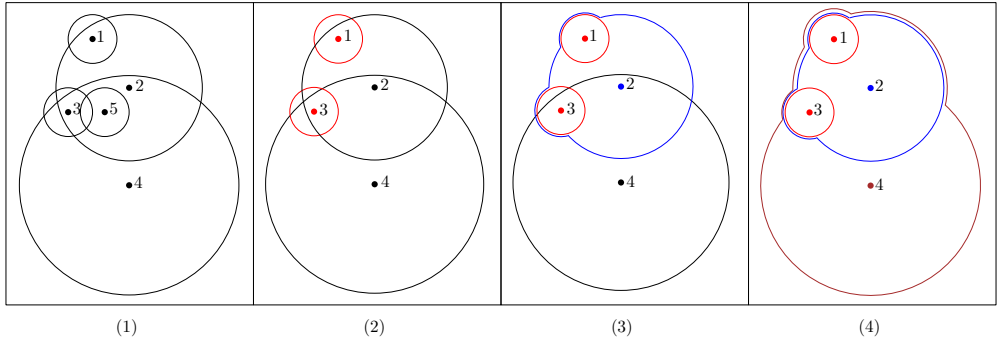}
    \caption{There are five facilities in the graph. (1): Before the filtering phase,
    3 and 5 are in conflict and 5 is filtered out. (2)-(4): We build the laminar family in
    non-decreasing order of $r_j$s.}
    \label{fig:fig}
    \end{center}
\end{figure*}

\paragraph{Building a laminar family for dangerous clients}
For any client $j \in D'$, let $B_j := \ball(j, \dmax(j)/15)$, where $\ball(j, L) = \set{i \in F: d(i,j) \leq L}$ is the set of facilities that are within a distance $L$ from $j$.
We notice that with the definition of $B_j$, if a copy of some facility $i$ is in $B_j$ (recall a facility may be split into several copies),
all copies of $i$ are in $B_j$.
We first present a few properties of $B_j$, then show how to construct the laminar family $\mB$.
The following lemma shows that the volume of $B_j$ is very close to $r_j$.

\begin{lemma}
\label{lemma:B-j-large}
For a client $j \in D$ with $r_j = r$, we have
\[
r - \frac{15\dav^r(j)}{\dmax(j)}\leq y(B_j) <  r.
\]
\end{lemma}
\begin{proof}
Notice that $\dmax(j)/15 \geq \dav^r(j) \geq \dmax^{r-1}(j)$; all clients in $F_j \setminus B_j$ contribute to $\dav^r(j)$. Thus we have $$\dav^r(j) \geq y(F_j \setminus B_j)\dmax(j)/15,$$ which implies
\begin{align*}
y(B_j) & = r - y(F_j \setminus B_j)  \geq r - \frac{\dav^r(j)}{\dmax(j)/15}.
\end{align*}
\end{proof}

In particular, Lemma~\ref{lemma:B-j-large} implies that $y(B_j) \geq r - 15/45 = r-1/3$.
The following lemma shows that two distinct dangerous clients in $D'$ are necessarily far way.
A corollary of the lemma which is useful later is that $B_j$ and $B_{j'}$ are disjoint.

\begin{lemma}
\label{lemma:dangerous-clients-far-away}
Let $j$ and $j'$ be two distinct clients in $D'$ such that $r_j = r_{j'} = r$. Then
\[
d(j, j') \geq \dmax(j)/10 + \dmax(j')/10.
\]
\end{lemma}

\begin{proof}
Assume otherwise. Then, by triangle inequalities,
\[\card{\dmax(j) - \dmax(j')} \leq d(j, j') < \dmax(j)/10 + \dmax(j')/10.\] Thus,
\[
\frac{\dmax^r(j')}{\dmax^r(j)} \in \left[\frac{1-1/10}{1+1/10}, \frac{1+1/10}{1-1/10}\right] = \left[\frac{9}{11}, \frac{11}{9}\right].
\]

Since $B_{j'} \subseteq \ball(j, d(j, j') + \dmax(j')/15)$ and
\ifdefined \stoc
\begin{align*}
d(j, j') + \dmax(j')/15 &\leq \frac{1}{10}\left(1+\frac{11}{9}\right)\dmax(j) + \frac{11/9}{15}\dmax(j)\\
&< 0.5\dmax(j),
\end{align*}
\else
\[d(j, j') + \dmax(j')/15 \leq \frac{1}{10}\left(1+\frac{11}{9}\right)\dmax(j) + \frac{11/9}{15}\dmax(j) < 0.5\dmax(j),\]
\fi
we have $B_{j'} \subseteq \ball(j, 0.5\dmax(j))$. Thus, we have $B_j \cup B_{j'} \subseteq \ball(j, 0.5\dmax(j))$, implying $y(B_j \cup B_{j'}) < r$, which further implies
\[
y(B_j \cap B_{j'}) = y(B_j) + y(B_{j'}) - y(B_j \cup B_{j'}) \geq r - \frac23.
\]
Then, $\dav(j, B_j \cap B_{j'}) \leq r\dav(j)/(r-2/3) \leq 3\dav(j)$. Similarly, $\dav(j', B_j \cap B_{j'}) \leq 3\dav(j')$. By triangle inequality $d(j, j') \leq 3(\dav(j) + \dav(j')) \leq 6\max\set{\dav(j), \dav(j')}$. $j$ and $j'$ can not be both in $D'$ since they conflict each other, leading to a contradiction.
\end{proof}

The following lemma shows that if two dangerous clients with different demands are close to each other, the ball for the client with the larger
demand is necessarily much larger than the one for the other client.

\begin{lemma}
\label{lemma:dangerous-geometrically-decrease}
Let $j$ and $j'$ be two clients in $D'$ with $r = r_j > r' = r_{j'}$. Suppose $d(j, j') \leq \dmax(j)/15 + \dmax(j')/10$. Then $\dmax(j') \leq \frac16\dmax(j)$.
\end{lemma}

\begin{proof}
Assume otherwise; then $\dmax(j) < 6\dmax(j')$.
Then, we have that
\begin{align*}
d(j, j') + \frac{\dmax(j)}{15} & \leq \frac{6\dmax(j')}{15} + \frac{\dmax(j')}{10} + \frac{6\dmax(j')}{15} \\
& =0.9\dmax(j')
\end{align*}
and $B_j \subseteq \ball\left(j', d(j, j') + \frac{\dmax(j)}{15}\right)$.
Thus, we have $B_j\subseteq \ball\left(j', 0.9\dmax(j')\right)$.
Since $y(B_j) \geq r - 1/3 > r-1 \geq r'$, we have $y(\ball(j', 0.9\dmax(j')) \geq r'$, contradicting the definition of $\dmax$.
 \end{proof}

In fact,
if $j$ and $j'$ satisfy the condition of Lemma~\ref{lemma:dangerous-geometrically-decrease},
we can see that the distance from every point in $B_{j'}$ to $j$ is at most
\ifdefined \stoc
\begin{align*}
d(j,j')+\frac{1}{15}\dmax(j') &\leq \frac{1}{15}\dmax(j)+\frac{1}{10}\dmax(j')+\frac{1}{15}\dmax(j')\\
&\leq (\frac{1}{15} +\frac{1}{36} )\dmax(j).
\end{align*}
\else
\[
d(j,j')+\frac{1}{15}\dmax(j') \leq \frac{1}{15}\dmax(j)+\frac{1}{10}\dmax(j')+\frac{1}{15}\dmax(j')\leq (\frac{1}{15} +\frac{1}{36} )\dmax(j).
\]
\fi
Intuitively, this suggests that $B_{j'}$ is {\em almost} contained in $B_j$.
If the condition of Lemma~\ref{lemma:dangerous-geometrically-decrease} does not hold, $j$ and $j'$
are obviously disjoint.
Therefore, we can see the family $\{B_{j}\}_{j\in D'}$ is {\em almost laminar}.
In fact, by slightly modifying the sets $B_{j}$,  we can form a laminar family.

Now, we present the algorithm for creating the laminar family $\mB$.
For any client $j \in D'$, we now construct a new set $B'_j \supseteq B_j$,
which is $B_j$ plus a small volume set of facilities.
Algorithm~\ref{alg:building-laminar-family} describes the process.
See Figure~\ref{fig:fig} for an illustration of our algorithm.
We prove that $\{B'_j\}_{j \in D'}$ forms a laminar family.

\begin{algorithm}
\caption{building a laminar family $\mB = \set{B'_j:j \in D'}$ of sets}
\label{alg:building-laminar-family}
\begin{algorithmic}[1]
\STATE \textbf{For} $r = 1$ to $R$ \textbf{do}
\STATE \hspace{\algorithmicindent} \textbf{For} each client $j \in D'$ such that $r_j = r$ \textbf{do}
\STATE \hspace{\algorithmicindent}\hspace{\algorithmicindent} Let $D''$ be the set of clients $j'$ such that $r_{j'} < r$ and $B'_{j'} \cap B_j \neq \emptyset$;
\STATE \hspace{\algorithmicindent}\hspace{\algorithmicindent}$B'_j \leftarrow B_j \cup \bigcup_{j' \in D''}B'_{j'}$; \label{STATE:create-Bp-j}
\end{algorithmic}
\end{algorithm}

\begin{lemma}
\label{lemma:laminar-family} The following properties hold for $\mB=\{B'_j\}_{j \in D'}$:
\begin{enumerate}
\item $B'_j \subseteq \ball(j, \dmax(j)/10)$ for every $j \in D'$;
\item $\mB=\{B'_j\}_{j \in D'}$ forms a laminar family.
\end{enumerate}
\end{lemma}

\begin{proof}
We prove  both the statements together by induction or $r$. We prove $B'_j \subseteq \ball(j, \dmax(j)/10)$ for any client $j$ such that $r_j \leq r$; also, the family $\mB_r=\{B'_j\}_{j \in D':r_j \leq r}$ form a laminar family.  If $r = 1$, we have $B'_j = B_j = \ball(j, \dmax(j)/15)$ for every $j \in D'$ with $r_j = 1$. Also, by Lemma~\ref{lemma:dangerous-clients-far-away}, $B'_j$ and $B'_{j'}$ are disjoint for two distinct clients $j$ and $j'$ in $D'$ with $r_j = r_{j'} = 1$.  Thus the statements are true for $r = 1$.

Suppose the statement is true for $r - 1$.  Consider two clients $j$ and $j'$ in $D'$ such that $r_j  = r, r_{j'} < r$ and $B_j \cap B'_{j'} \neq \emptyset$.  By the induction hypothesis, $B'_{j'} \subseteq \ball(j', \dmax(j')/10)$, implying $d(j, j') \leq \dmax(j)/15 + \dmax(j')/10$. By Lemma~\ref{lemma:dangerous-geometrically-decrease}, $\dmax(j') \leq \frac16\dmax(j)$.  Then, $d(j, j') + \dmax(j')/10 \leq \dmax(j)/15 + \dmax(j)/60 + \dmax(j')/60 = \dmax(j)/10$. Thus,
$$B'_{j'} \subseteq \ball(j, d(j, j') + \dmax(j')/10) \subseteq \ball(j, \dmax(j)/10).$$
This is true for any such client $j'$. By the definition of $B'_j$ at Line~\ref{STATE:create-Bp-j}, we have
that
\[
B'_j \subseteq \ball(j, \dmax(j)/10).
\]
Consider two distinct clients $j, j' \in D'$ such that $r_j = r_{j'} = r$.
We claim that there is no $j''$ such that $r_{j''} < r$ and $B'_{j''}$ intersect both $B_j$ and $B_{j'}$.
Assume there is such a client $j''$. Then, we have that
\begin{align*}
d(j, j'') & \leq \frac{\dmax(j)}{15} + \frac{\dmax(j'')}{10} \leq \frac{\dmax(j)}{12}.
\end{align*}
Similarly $d(j', j'') \leq \dmax(j')/12$. Thus, $d(j, j') \leq \dmax(j)/12 + \dmax(j')/12$.  Contradicting Lemma~\ref{lemma:dangerous-clients-far-away}.

Notice that in order to construct $B'_j$ at Line~\ref{STATE:create-Bp-j}, it is enough to consider
the sets in $\mB_{r-1}=\{B'_{j''} \mid j'' \in D' , r_{j''} \leq  r-1\}$ that are inclusively maximal
(those that are not properly contained by other set in $\mB_{r-1}$).
By the induction hypothesis, these inclusively maximal sets are disjoint.
Thus, for any clients $j, j' \in D'$ with $r_j = r_{j'} = r$,  $B'_j$ and $B'_{j'}$ are disjoint.
Moreover, for any $j'' \in D'$ with $r_{j''} < r$, either $B'_{j''} \subseteq B'_j$ or $B'_{j''} \cap B'_j  = \emptyset$.
Thus, the family $\mB_r=\{B'_j:j \in D', r_j \leq r\}$ is laminar.
\end{proof}

\subsection{Rounding}
\label{subset:rounding}
After obtaining a LP solution $\left(\set{y_i:i \in F}, \set{F_j : j \in C}\right)$, we run the algorithm of \cite{Yan12} as described in Section~\ref{subsec:review-yan} to obtain a family $\mU$ of disjoint bundles and the sets $\set{U_{j, t} : j \in C, t \in [r_j]}$. We then create the laminar family $\mB =
\{B'_j : j \in D'\}$ of sets.  Notice that by Lemma~\ref{lemma:laminar-family}, we have $\ball(j, \dmax(j)/15) = B_j \subseteq B'_j \subseteq \ball(j, \dmax(j)/10)$.  Thus, $r_j - 1 \leq y(B'_j) \leq r_j$.
Consider the polytope defined by the following set of constraints.
The set of variables is $\{z_i : i \in F\}$:
\begin{enumerate}
\item $\sum_{i \in U}z_i = 1 \quad\forall U \in \mU $
\item $r_j-1 \leq \sum_{i \in B'_j}z_i \leq r_j \quad\forall j \in D'$
\item $\sum_{i' \in g^{-1}(i)} z_{i'} \leq 1 \quad \forall i \in F $
\item $\sum_{i \in F}z_i  = k$
\end{enumerate}
From the construction of $B'_j$, it is easy to see that either $g^{-1}(i) \subseteq B'_j$ or $g^{-1}(i) \cap B'_j = \emptyset$
for any  $i\in F$ and  $j\in D'$.
Thus, $\mB \cup \set{F} \cup \set{g^{-1}(i):i \in F}$ forms a laminar family. The constraints of the above polytope is defined by two laminar families of sets : $\mU$ and $\mB \cup \set{F} \cup \set{g^{-1}(i):i \in F}$.
It is well known that such a polytope defined by two laminar families is integral.
Also, notice that the $z_i = y_i$ for every $i \in F$ is  a feasible solution.
Thus, we can express our vector $y$ as a convex combination of vertices of the above polytope.
Such a convex combination can be computed in polynomial time.
Treating the coefficients in the convex combination as probabilities (note that the coefficients sum up to $1$),
we sample a random vertex.
Due to the last constraint, the vertex contains exact $k$ open facilities.
Let $S$ be the set of $k$ facilities defined by the vertex.
We summarize the useful properties of our rounding step as follows.
\begin{enumerate}
\item The probability that each facility $i \in F$ is open is exactly $y_i$;
\item For any $i\in F$, we open at most 1 facility inside $g^{-1}(i)$;
\item We open exactly 1 facility inside each $U \in \mU$;
\item For each $j \in D'$, we open either $r_j-1$ or $r_j$ facilities in $B'_j$.
Moreover, we have that
\begin{eqnarray*}
\Prob[r_j \text{ facilities are open in }B'_j] &=& y(B'_j)-(r_j-1)\text{ and }\\
\Prob[r_j-1 \text{ facilities are open in }B'_j] &=& r_j- y(B'_j)
\end{eqnarray*}
\end{enumerate}

\subsection{Analysis}
\label{subsec:proof-of-constant}

We now have every piece ready to prove a constant factor approximation for \ftm.
Each of the following lemmas deals with one type of clients.
First, we consider safe clients.
\begin{lemma}
\label{lemma:bound-for-C-minus-D}
For any client $j \in C\setminus D$ with $r_j = r$, the expected connection cost of $j$ is at most $93r\dav(j)$.
\end{lemma}
\begin{proof}
Notice that we always open 1 facility inside $U_{j,t}$ for every $t \in [r]$. We connect $j$ to the $r$ facilities in $\bigcup_{t \in [r]}U_{j,t}$. Connecting $j$ to the facility in $U_{j,t}$ costs at most $2\dmax^t(j) + \dav^t(j)$ in expectation, by Lemma~\ref{lemma:close-to-bundles}. Thus, the expected connection cost of $j$ is at most
\ifdefined \stoc
\begin{align*}
&\ \ \ \sum_{t=1}^r\left(2\dmax^t(j) + \dav^t(j)\right) \\
&\leq 2\sum_{t=1}^{r-1}\dav^{t+1}(j) + 2\dmax(j) + \sum_{t=1}^r \dav^t(j)\\
&\leq 3r\dav(j) + 2\dmax(j) \\
&\leq 3r\dav(j) + 2\times 45\dav^r(j) \leq 93r\dav(j),
\end{align*}
\else
\begin{align*}
\sum_{t=1}^r\left(2\dmax^t(j) + \dav^t(j)\right) \leq 2\sum_{t=1}^{r-1}\dav^{t+1}(j) + 2\dmax(j) + \sum_{t=1}^r \dav^t(j)\\
\leq 3r\dav(j) + 2\dmax(j) \leq 3r\dav(j) + 2\times 45\dav^r(j) \leq 93r\dav(j),
\end{align*}
\fi
where the first inequality used the fact that $\dmax^t(j) \leq \dav^{t+1}(j)$.
\end{proof}

\begin{lemma}
\label{lemma:bound-for-Dp}
For any client $j \in D'$ with $r_j = r$, the expected connection cost of $j$ is at most $46r\dav(j)$.
\end{lemma}
\begin{proof}
Notice that by Lemma~\ref{lemma:close-to-bundles}, the distance from $j$ to its $r$-th closest open facility is always at most $3\dmax(j)$.   We can bound the expected connection cost of $j$ as follows. If there are $r_j$ open facilities inside $B'_j$, we connect $j$ to the $r$ open facilities; otherwise (they are $r - 1$ open facilities),  we connect $j$ to the $r-1$ open facilities in $B'_j$ and a $r$-th open facility outside $B'_j$ whose distance to $j$ can be bounded by $3\dmax(j)$.  Thus, the expected connection cost of $j$ is at most
\ifdefined \stoc
\begin{align*}
&\ \ \ \sum_{i \in B'_j}d(j, i) y_i + \Prob[r_j-1 \text{ facilities are open in }B'_j] \times 3\dmax(j)\\
&\leq r\dav(j) + 3(r - y(B_j))\dmax(j) \\
&\leq r\dav(j) + 3\times 15\dav^r(j) \leq 46r\dav(j),
\end{align*}
\else
\begin{align*}
&\sum_{i \in B'_j}d(j, i) y_i + \Prob[r_j-1 \text{ facilities are open in }B'_j] \times 3\dmax(j) \leq r\dav(j) + 3(r - y(B_j))\dmax(j) \\
&\leq r\dav(j) + 3\times 15\dav^r(j) \leq 46r\dav(j),
\end{align*}
\fi
where the second inequality follows from Lemma~\ref{lemma:B-j-large}.
\end{proof}

\begin{lemma}
\label{lemma:bound-for-D-minus-Dp}
For any client $j \in D \setminus D'$ with $r_j = r$, the expected connection cost of $j$ is at most $52r\dav(j)$.
\end{lemma}
\begin{proof}
There is a $j' \in D'$ such that $r_j = r_{j'} = r, \dav(j') \leq \dav(j)$ and $d(j, j') \leq 6\dav(j)$.  By Lemma~\ref{lemma:bound-for-Dp}, the expected connection cost of $j'$ is at most $46r\dav(j')$. By triangle inequality, the expected connection cost of $j$ is at most
$
46r\dav(j') + rd(j, j') \leq 46r\dav(j) + 6r\dav(j) = 52r\dav(j).
$
\end{proof}

Combining Lemma~\ref{lemma:bound-for-C-minus-D}, \ref{lemma:bound-for-Dp} and \ref{lemma:bound-for-D-minus-Dp}, the expected connection cost of any client $j \in C$ is at most $93r\dav(j)$, leading to a $93$-approximation for \ftm.

\section{{\large \ftm} on Paths and HSTs}

We first consider the case where
all the facilities and clients are on a line.

\begin{theorem}
For the non-uniform \ftm\ on a line metric, the problem can be solved exactly in polynomial time.
\end{theorem}

In fact,  all we need is to show the linear program \eqref{lp:ftm} has an integral optimal solution. Unlike in the usual case, we can not show that the polytope defined by the LP constraints is integral. In fact, the polytope is the same as that for the general NP-hard $k$-median problem, thus not integral. The integral optimum is due to the specialty of the cost coefficients, i.e., $\dist(i,j)$.

\begin{lemma}
\label{lm:line}
If $\dist(i,j)$s are defined by a line metric,
the linear program \eqref{lp:ftm} always has an integer optimal solution.
\end{lemma}
\begin{proof}
We show for any fractional optimal solution $(x_{i,j},y_i)$,
we can construct an integral solution with the same cost.
By the splitting trick
\footnote{
Consider facility $i$.
Let $J_l$ be the set of clients on the left side of $i$ and $J_r$ the set of clients
on the right side.
Consider the numbers $\{x_{i,j}\}_{j\in J_1}\cup \{y_i-x_{i,j}\}_{j\in J_2}$.
These numbers split the interval $[0,y_i]$ into several pieces, and for each piece,
we create a facility with fractional value equal to the length of that piece.
},
we can assume that $x_{i,j}=\{0, y_i\}$.
Each client (fractionally) connects to a consecutive segment of facilities.
Suppose $i$ is needed by demands set $J$.


Now we can write another linear program without $x_{i, j}$ variables as follows.
We use $i'$ for indexing the facilities after the split and $i$ for original facility.
We write $i'\in \mathsf{sp}(i)$ to indicate that the new facility $i'$ is derived from
the original facility $i$.
Let $F_j$ be the set of facilities serving $j$ (after the splitting process).
The facilities in $F_j$ form a consecutive segment in the path.

\begin{align}
\label{lp:new}
\textrm{minimize}\qquad & \sum_{j}\sum_{i'\in F_j}\dist(i', j)y_{i'}\\
\textrm{subject to}\qquad & \sum_{i'\in F_j}y_{i'} \ge r, \forall j \notag\\
& \sum_{i'\in \mathsf{sp}(i)}y_{i'} \leq 1, \forall i \notag \\
& \sum_{i'\in F}y_{i'} \leq k, \forall i \notag
\end{align}

It is easy to see that the optimal solution for the new LP is no more than that for the original LP.
The constraint matrix of the new LP has the consecutive ``one"s property:
in each row of the constraint matrix, the ``1"s appear in consecutive positions.
Such matrices are known to be totally unimodular
and the corresponding linear program has an integral optimal solution.
(See e.g.,\cite{schrijverbook}).
Furthermore, it is easy to see any integral feasible solution of \eqref{lp:new}
corresponds to a feasible solution for \ftm\ with the same cost.
Therefore, the optimal integral solution of \eqref{lp:new} has to be the same
as that of \eqref{lp:ftm}. The above argument also gives us an algorithm
to construct an integral solution of \eqref{lp:ftm} of the optimal cost.
\end{proof}

Using the same idea, we can get a polynomial time algorithm on
an HST metric where all facilities and clients are located at leaves.
We recall an HST (hierarchically well separated tree)
is a tree where on any root to leaf path, the edge lengths
decrease by some fixed factor in each step.

\begin{lemma}
\label{lm:hst}
The general \ftm\ problem can be solved exactly in polynomial time on
an HST metric where all facilities and clients are located at leaves.
\end{lemma}
\begin{proof}
We use $LCA(j_1, j_2)$ to denote the least common ancestor of leaves $j_1$ and $j_2$.
Suppose the leaves of the HST are ordered according the preorder traversal.
Consider a client $j$ and suppose the path from $j$ to the root is $\{j, p_1, p_2,\ldots, r\}$.
In a fractional optimal solution $(x_{i,j}, y_i)$ of \eqref{lp:ftm},
client $j$ chooses to connect all the facilities in the subtree rooted at $p_1$,
then those at $p_2$, and so on.
For any leaves $j_1,j_2,j_3$, if $LCA(j_1, j_2) = LCA(j_1, j_3)$,
we can easily see that $\dist_T(j_1, j_2) = \dist_T(j_1, j_3)$.
Therefore, we can assume $j$ connects to a consecutive segment of facilities
(in the preorder sequence of the facilities).
Using almost the same argument as in Lemma~\ref{lm:line},
we can show that the LP has an integral solution with the optimal value.
\end{proof}

Note that combining this result with classic tree embedding result~\cite{bartal1998approximating,JSK2003treeembed},
we can easily get a simple $O(\log{n})$-approximation for general \ftm\ on any metric.
Since we have already shown a constant approximation for general \ftm, we omit the details.

\ifdefined\fullpaper\else
\bibliographystyle{alpha}
\bibliography{cluster}
\end{document}
\fi

\ifdefined\fullpaper\else
\documentclass[11pt]{article}

\begin{document}
\fi

\section{Fault Tolerant Facility Location}
\label{sec:ftfl}

For \ftfl\ problem with arbitrary weights, we have a set $F$ of $n$ facilities
and a set $C$ of $m$ clients.
In the following sections, the terms ``demand" and ``client" are used interchangeably.
For each client $j$, there is
a nonnegative weight vector $\bw_j=\{w_j^{(1)},\ldots, w_j^{(r_j)}\}$
for some $r_j\leq n$.
Assume that the set of open facilities are $i_1, i_2, \ldots, i_h$ for some $1\leq h\leq n$,
sorted according to the nondecreasing order of their distance to $j$.
The service cost of client $j$ is $\sum_{t=1}^{r_j} w_j^{(t)} \dist(i_t, j)$.
If $h<r_j$, the service cost of $j$ is infinity.

We focus on a special case of the above problem where only one entry of the vector $\bw_j$ is nonzero. For ease of notation, we use $r_j$ to denote the index of the nonzero coordinate in $\bw_j$ and $w_j$ to denote $w_j^{(r_j)}$,  i.e., $w^{(r_j)}_{j} > 0$ and $w^{(t)}_{j}=0$ for any $t\ne r_j$.
Indeed, considering this special case is without loss of generality
since we can create multiple copies for each demand node $j$,
with the $1$st copy associated with the weight vector $\{w^{(1)}_j,0, \ldots, 0\}$,
the $2$nd copy $\{0, w^{(2)}_j, \ldots, 0\}$ and so on.
It is straightforward to establish the equivalence and we omit the proof here.
From now on, we use \ftfl\ to denote this special case of the fault
tolerant facility location problem.
Our main result is a constant factor approximation algorithm for \ftfl.

First, we note that the most natural linear integer programming formulation that was used for nonincreasing weight vectors in previous work does not work any more. 

Hence, we use a different linear integer programming formulation as follows.
We use boolean variable $y_i$ to denote whether facility $i$ is open,
$x_{ij}$ to denote whether demand $j$ is assigned to facility $i$.
We use $\pi(j,t)$ to denote the $t$th facility closest to $j$.
Let $N(j,t)=\{\pi(j,1), \pi(j,2),\ldots, \pi(j,t)\}$
and $c_{jt}=\dist(j, \pi(j,t))$.
Let $c_{j0}=0$ for all $j$.
We use indicator variable $z_{jt}$ to denote the event
whether demand $j$ is satisfied by
$N(j,t)$ (i.e., at least $r_j$ facilities among $N(j,t)$ are opened).
\begin{align}
\textrm{minimize} \ \ \ \ & \sum_{i}f_i y_i + \sum_j w_j \sum_{t\geq 0} (1-z_{jt}) (c_{j(t+1)}-c_{jt})  \label{lp:ftfl}
\end{align}
\ifdefined\stoc
\vspace{-10pt}
\fi
\begin{alignat}{2}
\textrm{s.t.}\ \ \ \  \sum_i x_{ij} &\geq r_j, &\quad \forall j\in C \label{serviceconst}\\
y_i &\geq x_{ij},  &\quad \forall i, j\in C \label{xyconst}\\
\sum_{i\in N(j,t)} x_{ij} &\geq r_jz_{jt}  &\quad \forall j\in C, \forall t\in [n]  \label{eq:gapconstraint} \\
y_i, x_{ij}, z_{jt} &\in \{0,1\}, &\quad\forall i\in F,  j\in C,  t\in [n]\cup\{0\}
\end{alignat}

First, we need to explain our objective function
since it is not the most frequently used objective
for facility location.
It is easy to see that a feasible solution of \ftfl\ satisfies the IP formulation.
For any optimal solution of the IP, if $N(j, t)$ satisfies $j$,
$N(j, t')$ also satisfies $j$ for $t'\geq t$.
Therefore, $z_{jt} \ge z_{j(t - 1)}$ for all $t$.
If $t'$ is the smallest $t$ such that $z_{jt}=1$, we can see that
$w_j\sum_{t\geq 0}(1 - z_{jt})(c_{j(t+1)}-c_{jt})$ is equal to
$w_jc_{jt'}$, which is exactly the service cost of $j$. We set $c_{j(n+1)}=\infty$.
Constraints \ref{serviceconst} specify that client $j$ must be connected to $r_{j}$ facilities.
Constraints \ref{xyconst} ensure that a client is connected only to open facilities and
constraints \ref{eq:gapconstraint} imply that if $z_{jt}=1$ then at least $r_j$ facilities
must be open in $N(j,t)$.
The LP relaxation is obtained by replacing last constraints
by $y_i, x_{ij}, z_{jt} \in [0,1]$.

However, we can not use the above LP directly to get a constant factor approximation algorithm
since its integrality gap is large and can be as large as $\Omega(n)$.
Consider the following \ftfl\ instance in a line metric.
There are $n$ facilities and only one client.
All facilities have cost zero and the client have demand $n$ (i.e., $r_1=n$).
The $x$-coordinate of the client is 0.
The $x$-coordinate of the $i$th facility is $0$ for all $1\leq i\leq n-1$
and the $x$-coordinate of the $n$th facility is $n$.
The optimal integral solution opens all facilities and the service cost is $n$.
A feasible fractional solution opens all facilities too.
However, $z_{jt}$ can take fractional values  $\frac{1}{n}\sum_{i\in N(j, t)}x_{ij}=\frac{t}{n}$.
The fractional service cost of the client is $\frac{n-1}{n}\cdot 0 +\ldots +\frac{2}{n}\cdot 0 +\frac{1}{n}\cdot n=1$.
Therefore, we obtain an integrality gap of $\Omega(n)$.

To strengthen the LP relaxation, we use the following {\em knapsack cover constraints}
to replace constraints (\ref{eq:gapconstraint}):
\begin{align}
\label{knapsackconst}
\sum_{i\in N(j,t)\setminus A} x_{ij}\geq (r_j-|A|)z_{jt} , \ \quad\ \ \forall j\in C,  t\in [n],  A\subseteq N(j,t)
\end{align}
The constraints require that if $z_{jt}=1$, then for every subset $A$, at least $r_j-|A|$ facilities
from the set $N(j,t)\setminus A$ must be chosen to serve $j$.
We can also see that there is a polynomial time separation oracle for \eqref{knapsackconst}:
Suppose $(x_{ij}, z_{jt})$ is a solution. For fixed $t$ and $j$,
we can test the feasibility of \eqref{knapsackconst} for all $A$ with $|A|=k$
by checking whether the sum of the smallest $|N(j,t)|-k$ terms in $N(j,t)$
is at least $(r_j-k)z_{jt}$.
Therefore, the relaxation can be solved optimally in polynomial time by the ellipsoid algorithm.
Let $(x^*,y^*,z^*)$ be the optimal fractional solution of the linear program
and $\opt$ be the optimal value.

\eat{
We first make two simple yet useful observations.

\begin{observation}
Without loss of generality, we can assume that,  for any demand $j$ and $1\leq t'\leq t\leq n$,
we have that $z_{jt'}\leq z_{jt}$.
\end{observation}

\begin{observation}
Without loss of generality, we can assume that, for any demand $j$ and $h>0$,
if $\sum_{i\in N(j,h)} y_i \leq r_j$, then $x_{hj}=y_h$.
\end{observation}
}
Now, we round the fractional solution $(x^*,y^*,z^*)$ to an integral solution $(\hx,\hy,\hz)$ as follows.
Let us consider a particular demand $j$.
Let $\alpha<1$ be a constant fixed later.
Let $t^*_j$ be the smallest integer $t$ such that $z^*_{jt}\geq \alpha$.

\begin{lemma}
\label{lm:cjt}
For every $j$, it holds that $c_{jt^*_j} \leq \frac{1}{1-\alpha}\sum_{t=0}^{n-1} (1-z^*_{jt}) (c_{j(t+1)}-c_{jt}).$
\end{lemma}

\begin{proof}
\begin{align*}
&\frac{1}{1-\alpha}\sum_{t=0}^{n-1} (1-z^*_{jt}) (c_{j(t+1)}-c_{jt}) \\
&\geq \frac{1}{1-\alpha}\sum_{t=0}^{t^*_j-1} (1-\alpha) (c_{j(t+1)}-c_{jt})  =c_{jt^*_j}.
\end{align*}
The first inequality follows because $z^*_{jt} \ge z^*_{j(t - 1)}$ for all $t$. This is true because if we set $z^*_{j,t}=max{z^*_{j,1},...,z^*_{j,t}}$, it yields a feasible solution of no greater cost.
\end{proof}


\noindent
Now, we create a set of $\ty_i$ values that we will round,
based on the $y^*_i$ values, as follows.

\begin{enumerate}
\item For all facility $i$ with $y^*_i\geq \alpha$, we round it up to $1$, i.e., $\ty_i=1$.
\item For all facility $i$ with $y^*_i< \alpha$, we let $\ty_i=\frac{1}{\alpha} y^*_i$.
\end{enumerate}

\begin{lemma}
For each client $j$, $\sum_{i\in N(j,t^*_j)} \ty_i \geq r_j$.
\end{lemma}
\begin{proof}
Consider a particular client $j$.
Let $A$ be the set of facility $i$ such that $x^*_{ij}\geq \alpha$ and $i\in N(j,t^*_j)$.
From \eqref{knapsackconst}, we know that
\ifdefined \stoc
\begin{align*}
&\ \sum_{i\in N(j,t^*_j)\setminus A} y^*_{i}\,\,\,\geq \,\,\,
\sum_{i\in N(j,t^*_j)\setminus A} x^*_{ij}\,\,\\
&\geq\,\, z^*_{jt^*_j}(r_j-|A|) \,\,
\geq\,\, \alpha (r_j-|A|).
\end{align*}
\else
\[
\sum_{i\in N(j,t^*_j)\setminus A} y^*_{i}\,\,\,\geq \,\,\,
\sum_{i\in N(j,t^*_j)\setminus A} x^*_{ij}\,\,\geq\,\, z^*_{jt^*_j}(r_j-|A|) \,\,\geq\,\, \alpha (r_j-|A|).
\]
\fi
Therefore, we can see that
$$
\sum_{i\in N(j,t^*_j)\setminus A} \ty_{i} \,\, \geq\,\,
\sum_{i\in N(j,t^*_j)\setminus A} \frac{1}{\alpha} y^*_{i}
\,\, \geq \,\,r_j-|A|.
$$
For each facility $i\in A$, we have $\ty_{i}=1$. Hence,
$
\sum_{i\in N(j,t^*_j)} \ty_{i}\geq r_j,
$
which completes the proof.
\end{proof}

Now, we round the $\ty$ values to integers.
Our rounding scheme is a slight variant of the one in \cite{swamy2008fault}.
Let $F_j= N(j, t^*_j)$.
Let $r'_j$ be the residual requirement of $j$, which is initially set to be $r_j$.
We iterate the following steps until no client remains in the graph.

\begin{itemize}
\item[S1.] We pick the client $j$ with the minimum $c_{jt^*_j}$.
\item[S2.] Let $M\subseteq F_j$ be the set of the cheapest facilities in $F_j$ (w.r.t. facility opening costs)
such that $\sum_{i\in M} \ty_i \geq r'_j$.
If $\sum_{i\in M}\ty_i$ is strictly large than $r'_j$, we replace the last facility, say facility $i$,
by two ``clones" $i_1$ and $i_2$.
Set $\ty_{i_1}= r'_j-\sum_{i\in M\setminus \{i\}} \ty_i$
and
$\ty_{i_2}=\ty_i-\ty_{i_1}$.
Include $i_1$ in $M$. Hence, $\sum_{i=M}y_i=r'_j$.
\item[S3.] Open the $r'_j$ cheapest facilities in $M$.
For each client $k$ with $F_k\cap M\ne \emptyset$,
we use any $\min(r'_k, r'_j)$ of the facilities we just opened to serve $k$
and let $r'_k=r'_k-\min(r'_k, r'_j)$.
Delete facilities in $M$
and all clients with zero residual requirement from the input.
\end{itemize}

\begin{lemma}
\label{lm:rounding}
The above rounding scheme returns a feasible solution. Moreover,
the following properties hold.
\begin{enumerate}
\item The facility opening cost is at most $\sum_i f_i \ty_i$.
\item For each client $j$, at least $r_j$
facilities in $B(j, 3c_{jt^*_j} )$ are open.
\end{enumerate}
\end{lemma}
\begin{proof}
The proof is almost the same as the one in \cite{swamy2008fault}.
For completeness, we include it here.
Consider a particular iteration.
It is easy to see the invariant $\sum_{i\in F_j}\ty_i\geq r'_j$ is maintained
throughout the three steps. So it is always possible to choose the set $M$.
We also need to argue that no facility is opened twice since
we have made some clones.
We argue that whenever a facility $i$ is replaced by two clones,
the first clone never gets opened:
This is simply because $i$ is the most expensive facility in $M$
and there are at least $r'_j$ facilities cheaper than $i$
(otherwise, we do not have to make clones).

To bound the facility cost, just notice that
the cost of open facilities in $M$ is less than $\sum_{i\in M} f_i \ty_i$.
This proves (1).
To bound the connection cost, consider a particular client $j$.
Any opened facility in $F_j$ is at most $c_{jt^*_j}$ distance away from $j$.
Notice that $j$ may be served by some facilities in $F_{k}$ for some other client $k$.
This only happens if $F_j\cap F_{k}\ne \emptyset$ and $c_{kt^*_k}\leq c_{jt^*_j}$
(we process client $k$ first).
A facility in $F_k$ is at most $2c_{kt^*_k}+ c_{jt^*_j}\leq 3c_{jt^*_j}$ away from $j$.
\end{proof}

From Lemma~\ref{lm:rounding}, we know
that the first $r_j$ copies of client $j$ are assigned within a distance of $3 c_{jt_j^*}$.
Therefore, we have that the total cost of this integral solution
\begin{align*}
\sol &\leq \frac{1}{\alpha} \sum_i f_i y^*_i + 3\sum_j w_j c_{j(t^*_j)}\\
&\leq \frac{1}{\alpha} \sum_i f_i y^*_i + \frac{3}{1-\alpha}  \sum_j w_j \sum_t (1-z^*_{jt}) (c_{j(t+1)}-c_{jt})
\end{align*}
where the second inequality holds because of Lemma~\ref{lm:cjt}.

Setting $\alpha=\frac{1}{4}$ gives us an approximation ratio of $4$.
We can choose a random $\alpha$ to improve the approximation ratio as in \cite{shmoys1997approximation,guha2003constant}.
Let $L_j(\alpha)$
be $c_{jt}$ for the minimal $t$ such that $z_{jt}>\alpha$.
It is easy to see the following.
\begin{lemma}
$$
\int_{0}^{1}L_j(\alpha) \d\alpha = \sum_t (1-z^*_{jt}) (c_{j(t+1)}-c_{jt}).
$$
\end{lemma}
Choose a random $\alpha$ uniformly distributed over $[h,1]$.
Then, the expected cost is
\ifdefined \stoc
\begin{align*}
\Exp[\sol]  \,\,\leq\,\, &
\int_{h}^{1} \frac{1}{1-h}\Bigl( \frac{1}{\alpha} \sum_i f_i y^*_i + 3\sum_j w_j L_j(\alpha) \Bigr)\d \alpha \\
\,\,\leq\,\, &\frac{1}{1-h} \ln \frac{1}{h} \sum_i f_i y^*_i \\
&\ \ + \frac{3}{1-h}  \sum_j w_j \sum_t (1-z^*_{jt}) (c_{j(t+1)}-c_{jt})
\end{align*}
\else
\begin{align*}
\Exp[\sol]  \,\,\leq\,\, &
\int_{h}^{1} \frac{1}{1-h}\Bigl( \frac{1}{\alpha} \sum_i f_i y^*_i + 3\sum_j w_j L_j(\alpha) \Bigr)\d \alpha \\
\,\,\leq\,\, &\frac{1}{1-h} \ln \frac{1}{h} \sum_i f_i y^*_i
+ \frac{3}{1-h}  \sum_j w_j \sum_t (1-z^*_{jt}) (c_{j(t+1)}-c_{jt})
\end{align*}
\fi
The above expression is minimized at $h=e^{-3}$,
which gives an approximation ratio $3.16$.

\begin{theorem}
There is a polynomial time approximation approximation
with an approximation factor $3.16$ for \ftfl.
\end{theorem}

\ifdefined\fullpaper\else
\bibliographystyle{alpha}
\bibliography{cluster}
\end{document}
\fi

\bibliographystyle{plain}
\bibliography{cluster}

\end{document}